\documentclass[12pt]{article}
\usepackage{amssymb,amscd,amsmath,amsthm}
\usepackage[colorinlistoftodos]{todonotes}
\usepackage[colorlinks=true, allcolors=blue]{hyperref}
\usepackage{lmodern}
\usepackage[english]{babel}
\usepackage{color}
\usepackage{enumitem}
\usepackage{graphicx}
\usepackage{tikz}
\usepackage[margin=1in]{geometry}
\newtheorem{theorem}{Theorem}[section]

\newtheorem{corollary}[theorem]{Corollary}
\newtheorem{lemma}[theorem]{Lemma}
\newtheorem{definition}[theorem]{Definition}

\newtheorem{remark}[theorem]{Remark}
\def\Tr{\mathrm{Tr}}
\def\id{{\bf 1}\!\!{\rm I}}

\date{January 2023}
\begin{document}

\centerline{\large{\textbf{Bi- Entangled Hidden  Markov Processes and Recurrence}}}
\bigskip\bigskip
\centerline{\textbf{Soueidi El Gheteb}}
\centerline{Department of Mathematics and Informatics, Faculty of Sciences and Technologies,}
\centerline{University of Nouakchott, Nouakchott, Mauritania,}
\centerline{elkotobmedsalem@gmail.com, elkotobmedsalem@fst.e-una.mr}

\begin{abstract}
In this paper, we introduce the notion of Bi-entangled hidden Markov processes. These are  hidden quantum processes where the hidden processes themselves exhibit entangled Markov process, and the observable processes also exhibit entanglement. We present a specific formula for the joint expectation of these processes. Furthermore, we discuss the recurrence of the underlying quantum Markov processes associated to the Bi- entangled hidden Markov processes and we establish that, by restricting them within suitable commutative subalgebras (diagonal subalgebras) leads to the recovery of Markov processes defined by the hidden stochastic matrix. In this paper we
only deal with processes with an at most countable state space.

\end{abstract}

\section{Introduction}

Hidden Markov models (HMMs) are a  statistical model that was first proposed by Baum L.E. (Baum and Petrie, 1966) \cite{BaumPetr66}. They introduced the Baum-Welch algorithm, alternatively recognized as the Expectation-Maximization (EM) algorithm. This algorithm holds fundamental significance as it enables the acquisition of Hidden Markov Model parameters from observed data. Nowadays, HMMs are widely utilized in a multitude of practical scenarios due to their inherent flexibility. They have found application in various domains, including: speech recognition \cite{JelBahMer75}, \cite{RabLeSo83}, \cite{Rab89}, \cite{HuaYasMerv90},  bio-informatics \cite{Eddy98}, finance \cite{HasNat05}, \cite{Nguyet18}, genetics \cite{LiSte03}, machine learning \cite{GhahrJord97} and many other fields (see \cite{AccSegLuSs} for a wider discussion of the literature on HMMs). This demonstrates the broad range of fields where Hidden Markov Models are employed to address a wide array of practical challenges.

In \cite{AccSegLuSs}, we have broadened the scope of the traditional HMMs by introducing a more encompassing category called Hidden Processes (HPs). In the context of Hidden Processes, it is not obligatory for the underlying or hidden process to exhibit Markovian behavior. For example, it can encompass hidden Markov processes that do not necessarily conform to the Markov property themselves.

In the past few years, there have been developments in extending the definition of HMPs to include quantum aspects. You can find examples of these extensions in references such as \cite{WiesnCrutc08}, \cite{MonrWiesn11}, and \cite{SatGuru93}. However, it's worth noting that these extensions are limited in scope and primarily address specific elements of HMPs, such as dynamics or various statistical algorithms.

In \cite{AccSegLuSs}, we introduced the class of "hidden quantum processes" (HQP) along with its subset, "hidden quantum Markov processes" (HQMP). We established that this class of processes extends the concept of Quantum Markov Chains (QMC) in a manner similar to how classical hidden Markov processes generalize Classical Markov Chains. Furthermore, in \cite{SsSe}, we explored a subset of HQMMs known as "entangled Hidden Markov models" (EHMMs). These EHMMs are characterized by having an underlying process represented by an entangled Markov chain, which is the quantum extension of the notion of classical random walk as detailed in references (\cite{AcFi03-EMC}–\cite{SsSeBa}).

In the present paper, we  continue the analysis of nontrivial examples of
 hidden quantum Markov processes,  we introduce a new class of HQMMs, which we refer to as "Bi-entangled Hidden Markov Models"(Bi-EHMMs).  Both the hidden and observable processes exhibit entanglement, as indicated in references (\cite{AcFi03-EMC}–\cite{SsSeBa}). The bi-entanglement arises from utilizing an entangled transition expectation related to the underlying process, along with an entangled emission operator associated with the observable process. Our work involves establishing a structural theorem for Bi-EHMMs. In addition, we suggest a definition and elucidate criteria to determine the recurrence of the underlying process represented by an entangled Markov chain and the diagonal restriction of the underlying process is discussed.

Let us briefly mention the organization of the paper. After preliminary information (see Sect. 2), in Section 3, we introduce Bi-entangled hidden Markov models and state the structure theorem. The recurrence for the underlying quantum process is the subject of Sections 4 and 5, where we provide an explicit definition and provide a recurrence criteria. Section 6 focuses on the analysis of restricting the underlying quantum process to its diagonal elements.

\section{Preliminaries  on Hidden quantum Markov processes }


Let $H$ be a separable hilbert space and $\mathcal{B}_{H}:=\mathcal{B}(H)$ the algebra of bounded linear operator on $H$. We restrict ourselves to the case where $dim(H)=d<\infty$ and we will use the notation
$$
D := \{1,\dots, d\}
$$
To each ortho--normal basis (o.n.b.) $e\equiv(e_{h})_{h\in D}$ of $H$, one can associate a system of
matrix units $(e_{h,k})$ where, for all $h, k\in D$,
\begin{equation}\label{notat-ee^*}
e_{hk}(\xi) :=
e_{k}e_{h}^*(\xi) := \langle e_{h}, \xi\rangle e_{k}  \quad,\quad\forall \xi\in H
\end{equation}
and its matrix, in the $e$--basis, has all entries equal to zero with the exception of the $(h,k)$--th element which is 1. The assignement of the system of matrix units $(e_{hk})_{h,k\in D}$ allows to identify $\mathcal{B}_H$ with the algebra $\mathcal{M}_d (\mathbb{C})$ of all $d\times d$ complex matrices.\\
Let $d_H$ and $d_O$ be two  positive integers and let
$$
D_H: =\{1,\cdots, d_H\}\quad \hbox{and} \quad D_O =\{1,\cdots, d_O\}
$$
Define the hidden sample algebra by
\begin{equation}\label{sample_alg}
\mathcal{A}_{H}= \bigotimes_{\mathbb{N}}\mathcal{M}_{d_H}
\end{equation}

 and the observable sample algebra by
\begin{equation}\label{obser_alg}
\mathcal{A}_{O}=\bigotimes_{\mathbb{N}} \mathcal{M}_{d_O}
\end{equation}

and the $(H,O)$ full sample algebra by
\begin{equation}\label{(H,O)_sample_alg}
\mathcal{A}_{H,O}:=\bigotimes_{\mathbb{N}} (\mathcal{M}_{d_H} \otimes \mathcal{M}_{d_O})
\end{equation}

The corresponding tensor embedding ($n\in \mathbb{N}$)
\begin{equation}\label{df-H-embeds-cl}
\mathcal{A}_{H,n}:=j_{H_{n}}(\mathcal{M}_{d_H}) = \id_{d_H}\otimes \id_{d_H}\otimes\cdots \otimes \id_{d_H}\otimes \mathcal{M}_{d_H}\otimes \id_{d_H}\cdots
\end{equation}
\begin{equation}\label{df-O-embeds-cl}
\mathcal{A}_{O,n}:= j_{O_{n}}(\mathcal{M}_{d_O}) = \id_{d_O}\otimes \id_{d_O}\otimes\cdots \otimes \id_{d_O}\otimes \mathcal{M}_{d_O}\otimes \id_{d_O}\cdots
\end{equation}

The $(H,O)$-embedding is defined by
\begin{equation}\label{df-H-embeds-cl-alg}
\mathcal{A}_{H,n}\otimes \mathcal{A}_{O,n} :=j_{H_{n}}(\mathcal{M}_{d_H})\otimes j_{O_{n}}(\mathcal{M}_{d_O})
\subset \mathcal{A}_{H,O}
\end{equation}
Denote the finite volume $(H,O)$ sample algebra by
$$
\mathcal{A}_{H, O; k} = j_{H_k}(\mathcal{M}_{d_H}\otimes \mathcal{M}_{d_O});\quad \mathcal{A}_{H,O, [0,n]}:=\bigotimes_{k\in[0,n] }\mathcal{A}_{H, O, k}
$$
-- The backward (resp. forward) H-filtration is given by
$$
 \mathcal{A}_{H, [0,n]}:=\bigotimes_{k\in[0,n] }\mathcal{A}_{H, k}\quad ; \quad (\hbox{resp.}\, \mathcal{A}_{H, [n}:=\bigotimes_{k\ge n}\mathcal{A}_{H, k} )
$$
-- The finite volume O-sample algebra
$$
  \mathcal{A}_{O, [0,n]}:=\bigotimes_{k\in[0,n] }\mathcal{A}_{O, k}
$$
\begin{definition}
A linear map $\mathcal{E}_H$ from $\mathcal{M}_{d_H}\otimes \mathcal{M}_{d_H}$ into $\mathcal{M}_{d_H}$ is called \textit{transition expectation} if it is completely positive and  identity preserving.
\end{definition}
\begin{definition}
A linear map $\mathcal{E}_{H,O}$ from $\mathcal{M}_{d_H}\otimes \mathcal{M}_{d_O}$ into $\mathcal{M}_{d_H}$ is called \textit{emission operator} if it is completely positive and  identity preserving.
\end{definition}

\begin{definition}\label{def-QHMP}
A state $\varphi_{H,O}$ over $\mathcal{A}_{H,O}$ is a (homogeneous)  hidden quantum Markov processes (HQMP for
short) with:
\begin{enumerate}
  \item initial state $\varphi_{H,0}$ on $\mathcal{M}_{d_H}$;
  \item a transition expectation $\mathcal{E}_H:$ $\mathcal{M}_{d_H}\otimes \mathcal{M}_{d_H} \rightarrow \mathcal{M}_{d_H}$;
   \item \textit{an emission operator}
$\mathcal{E}_{H,O}:\, \mathcal{M}_{d_H}\otimes \mathcal{M}_{d_O} \rightarrow \mathcal{M}_{d_H}$;
\end{enumerate}
and
\begin{equation}\label{joint-exp-qHMP}
\varphi_{H,O}(\prod_{m=0}^{n}j_{H_m}(a_m)j_{O_m}(b_m))
\end{equation}
$$
= \varphi_{H_0}\left(\mathcal{E}_{H}\left(\mathcal{E}_{O,H} (a_{0}\otimes b_{0})\otimes\mathcal{E}_{H} (\mathcal{E}_{O,H} \left((a_{1}\otimes b_{1})\otimes \cdots \right.\right.\right. \notag\\
$$
$$
\left.\left.\left.\otimes \mathcal{E}_{H}\left( \mathcal{E}_{O,H}(a_{n-1}\otimes b_{n-1})\otimes
\mathcal{E}_{H}\left(\mathcal{E}_{O,H}\left(a_{n}\otimes b_{n}\right)\otimes 1_{H_{n+1}}
\right)\right)\cdots\right)\right)\right)\notag
$$
for all $n\in\mathbb{N}$,  for all $a_m\in\mathcal{M}_{d_H},b_m\in \mathcal{M}_{d_O}, \,m\in \{0, \dots, n\}$.\\
Then the triplet $(\varphi_{H_0}, \mathcal{E}_H, \mathcal{E}_{H,O})$ will be refereed as \textit{Hidden Quantum Markov Model} and the state given by (\ref{joint-exp-qHMP}) represents its joint expectation.

\end{definition}
In the above definition the pair  $(\varphi_{H_0},  \mathcal{E}_H)$ define a quantum Markov chain $\varphi_H$ on $\mathcal{A}_H$, whose joint expectation is given by
\begin{equation}\label{phiH}
  \varphi_H(\prod_{m=0}^{n}j_{H_m}(a_m)) = \varphi_{H_0}(\mathcal{E}_{H}(a_0\otimes\mathcal{E}_{H}(a_2\cdots \mathcal{E}_{H}(a_n\otimes \id_{H_{n+1}})\cdots )))
\end{equation}

\begin{remark}
By restricting $\varphi_{H,O}$ to an abelian sub-algebra of the $\mathcal{A}_{H,O}$ algebra, we obtain classical hidden Markov processes \cite{AccSegLuSs}.

\end{remark}
\begin{remark}
The \textit{emission operator} $\mathcal{E}_{H,O}$ characterizes the probabilities conditioned on the hidden process for the observable process. This hidden process is defined by its own transition expectation denoted as $\mathcal{E}_{H}$.
\end{remark}

\section{Bi- Entangled Hidden Markov Models}

\noindent In what follows, we make the assumption that  all the algebras $\mathcal{A}_{H_{n}}$ and $\mathcal{A}_{O_{n}}$ are taken to be
isomorphic to a single algebra $\mathcal{B}$, \textbf{independent of} $n$
and isomorphic to  the C$^{*}$-algebra of $d\times d$-matrices $\mathcal{M}_d$, its identity will be denoted by $\id_d$, for some $d\in \mathbb{N}$:
\begin{equation}\label{hid-obs-q-algs-equal}
\mathcal{A}_{H_{n}} \equiv \mathcal{A}_{O_{n}} \equiv \mathcal{B}:=\mathcal{M}_d
\end{equation}
where $\equiv$ denotes $*$--isomorphism.
Thus, for the sample algebra of the underlying Markov process, one has the identification
$$
\mathcal{A}_{H} := \bigotimes_{\mathbb{N}}\mathcal{B}
$$
with the tensor embeddings
\begin{equation}\label{df-jHn-calA-Hn}
j_{H_{n}} \colon b\in \mathcal{B} \to j_{H_{n}}(b) \equiv b\otimes 1_{\{n\}^c}\in\mathcal{A}_{H}
\quad;\quad \mathcal{A}_{H_{n}} := j_{H_{n}}(\mathcal{B})
\end{equation}
where $1_{\{n\}^c}$ denotes the identity in $\bigotimes_{\mathbb{N}\setminus\{n\}}\mathcal{B}$.\\
Similarly we define the observable algebra
$$
\mathcal{A}_{O} \equiv \bigotimes_{\mathbb{N}}\mathcal{B}
$$
and the tensor embeddings
\begin{equation}\label{df-jOn-calA-On}
j_{O_{n}} \colon b\in \mathcal{B} \to j_{O_{n}}(b) \equiv b\otimes 1_{\{n\}^c}\in\mathcal{A}_{O}
\quad;\quad \mathcal{A}_{O_{n}} := j_{O_{n}}(\mathcal{B})
\end{equation}
where again $1_{\{n\}^c}$ denotes the identity in
$\bigotimes_{\mathbb{N}\setminus\{n\}}\mathcal{B}$.\\
The algebra of the $(H,O)$--process is then
\begin{equation}\label{identif-calA(H,O)}
\mathcal{A}_{H,O} := \mathcal{A}_{H_{n}} \otimes \mathcal{A}_{O_{n}}
\equiv\bigotimes_{\mathbb{N}} \ (\mathcal{A}_{H} \otimes \mathcal{A}_{O})
\equiv\bigotimes_{\mathbb{N}}\mathcal{B}\otimes\bigotimes_{\mathbb{N}}\mathcal{B}
\end{equation}

Consider hidden Markov model $\lambda =(\pi, \Pi, Q)$, where
$\pi =(\pi_j)_{j\in D}$ the initial distribution of the hidden process, $\Pi=(\Pi_{ij})_{i,j\in D}
$ is the hidden stochastic matrix and  $Q=(q_{j}(k))_{\substack{j,k\in D}}
$ is the emission stochastic  matrix.\\
\begin{definition}
The entangled hidden Markov operator,denoted as $P_H$, and the entangled emission Markov operator, denoted as $P_{H,O}$, are defined in the following manner, and they respectively relate to the stochastic matrix $\Pi$ and the stochastic matrix $Q$, along with the canonical systems of matrix units  $\{e_{i,j}\}_{i,j \in D}$

\begin{equation}\label{Ph}
  P_H(A) = \sum_{i,j,k,l\in D}\sqrt{\Pi_{ik}\Pi_{jl}}\ a_{kl}e_{ij},\quad \forall A = (a_{kl})\in \mathcal{M}_{d}
\end{equation}
\begin{equation}\label{Ph_O}
  P_{H,O}(B) = \sum_{i,j,k,l\in D}\sqrt{Q_{ik}Q_{jl}}\ b_{kl}e_{ij}, \quad \forall B = (b_{kl})\in \mathcal{M}_{d}
\end{equation}
\end{definition}

\begin{remark}
One can immediately check that  $P_H$ and $P_{H,O}$ do not preserve identity. As a consequence, they exhibit entanglement according to the criteria defined in \cite{AcFi03-EMC}.
\end{remark}
Let $\mathcal{E}_{H}: \mathcal{M}_{d}\otimes \mathcal{M}_{d}\rightarrow \mathcal{M}_{d}$ be defined  as the linear extension of
\begin{equation}\label{E_H}
\mathcal{E}_{H}(a\otimes b) := a\diamond P_{H}(b), \quad a, b  \in \mathcal{M}_{d}
\end{equation}
and let
$\mathcal{E}_{H,O}: \mathcal{M}_{d}\otimes \mathcal{M}_{d}\rightarrow \mathcal{M}_{d}$
be defined as the linear extension of
\begin{equation}\label{E_H_O}
\mathcal{E}_{H,O}(a\otimes b) := a\diamond P_{H,O}(b), \quad a,b \in \mathcal{M}_{d}
\end{equation}
where $\diamond$ represent the \textbf{Schur} product (see \cite{AcFi03-EMC}).
\begin{lemma}
 $\mathcal{E}_{H}$ \eqref{E_H} and $\mathcal{E}_{H,O}$ \eqref{E_H_O} are completely positive and identity-preserving maps.\\
Moreover, for $a=(a_{ij})_{i,j\in D}$ and $b=(b_{kl})_{k,l\in D}$, we have
$$
\mathcal{E}_{H}(a\otimes b)=\sum_{i,j,k,l\in D}\sqrt{\Pi_{ik}\Pi_{jl}}a_{ij}b_{kl}e_{ij};\quad \mathcal{E}_{H,O}(a\otimes b)=\sum_{i,j,k,l\in D}\sqrt{Q_{ik}Q_{jl}}a_{ij}b_{kl}e_{ij}
$$
\end{lemma}

\begin{definition}
The mappings $\mathcal{E}_{H}$ (given by equation \eqref{E_H}) and $\mathcal{E}_{H,O}$ (given by equation \eqref{E_H_O}) are referred to as the entangled hidden transition expectation and the entangled emission operator, respectively.
\end{definition}
\begin{definition}
An HQMM $(\varphi_{H,0}, \mathcal{E}_{H}, \mathcal{E}_{H,O})$ is called \textit{Bi-entangled hidden Markov model}  if its hidden transition expectation and emission operator take the forms \eqref{E_H} and \eqref{E_H_O}, respectively.
\end{definition}
\begin{remark}
In the given definition, the entangled hidden Markov models \cite{SsSeg} can be identified when $\mathcal{E}_{H,O}$ is taken in the following form:
\begin{equation}\label{emission_operator_CDA}
\mathcal{E}_{H,O}(x):=\overline{\hbox{Tr}}_{2}(K_{H,O}^*x K_{H,O}), \quad x\in \mathcal{M}_{d}\otimes \mathcal{M}_{d}
\end{equation}
here $K_{H,O}\in \mathcal{M}_{d}\otimes \mathcal{M}_{d}$ is a conditional density amplitude. $\overline{\Tr}_2$ is the partial trace from $\mathcal{M}_{d}\otimes \mathcal{M}_{d}$ into $\mathcal{M}_{d}$ defined by linear extension of $\overline{\Tr}_{2}(a\otimes b) = a \Tr(b)$.

\end{remark}

\begin{lemma}\label{lem_jointH} Let  $n\in \mathbb{N}$. For $\mathcal{E}_H$ given by (\ref{E_H}),  we have that
\begin{equation}\label{jnt-exp}
\mathcal{E}_H(a_n\otimes (\mathcal{E}_H(a_{n+1}\otimes\cdots\otimes (\mathcal{E}_H(a_{n+r} \otimes \id_{d})))))=\sum_{\substack{k_n,\cdots, k_{n+r-1}\\l_n,\cdots, l_{n +r-1}}}\Big(\prod_{m =n}^{n+r-1}\sqrt{\Pi_{k_{m}k_{m+1}}\Pi_{l_{m}l_{m+1}}} \, a^{(m)}_{k_{m}l_{m}}\Big)
\end{equation}
$$
\times\Big(a^{(n+r)}_{k_{n+r}l_{n+r}}\sum_{j}\sqrt{\Pi_{k_{n+r}j}\Pi_{l_{n+r}j}}\Big)e_{k_nl_n}
$$

for each $r\ge 0$ and $a_k =(a^{(m)}_{k_ml_m})_{k_m,l_m\in D}\in \mathcal{M}_{d},\, n\le m \le n + r$.

\end{lemma}
\begin{proof} From (\ref{Ph}) we have
  $$P_{H}(\id_{d})= \sum_{k,l,i,j}\sqrt{\Pi_{ki}\Pi_{lj}}\delta_{i,j}e_{kl} =  \sum_{k,l,j}\sqrt{\Pi_{kj}\Pi_{lj}}e_{kl}$$
Then
$$
\mathcal{E}_H(a_{n+r}\otimes \id_{d}) = a_{n+r}\diamond P_H(\id_{d}) =  \sum_{k_{n+r},l_{n+r},j}\sqrt{\Pi_{k_{n+r}j}\Pi_{l_{n+r}j}}a^{(n+r)}_{k_{n+r}l_{n+r}}e_{k_{n+r}l_{n+r}}
$$
It follows that
\begin{eqnarray*}
 && \mathcal{E}_{H}(a_{n+r-1}\otimes \mathcal{E}_H(a_{n+r}\otimes \id_{d}))\\
   &=& a_{n+r-1}\diamond  P_{H}\Big(\mathcal{E}_H(a_{n+r}\otimes \id_{d})\Big)\\
    &=&  a_{n+r-1}\diamond  \sum_{\substack{k_{n+r-1},k_{n+r}, j\\l_{n+r-1}, l_{n+r}}} \sqrt{\Pi_{k_{n+r-1}k_{n+r}}\Pi_{l_{n+r-1}l_{n+r}}}\sqrt{\Pi_{k_{n+r}j}\Pi_{l_{n+r}j}}a^{(n+r)}_{k_{n+r}l_{n+r}}
    e_{k_{n+r-1}l_{n+r-1}}\\
 &=&  \sum_{\substack{k_{n+r-1},k_{n+r}, j\\l_{n+r-1}, l_{n+r}}} \sqrt{\Pi_{k_{n+r-1}k_{n+r}}\Pi_{l_{n+r-1}l_{n+r}}}\sqrt{\Pi_{k_{n+r}j}\Pi_{l_{n+r}j}}a^{(n+r-1)}_{k_{n+r-1}l_{n+r-1}}a^{(n+r)}_{k_{n+r}l_{n+r}}
    e_{k_{n+r-1}l_{n+r-1}}\\
\end{eqnarray*}
Then, the formula \eqref{jnt-exp} follows by iteration.
\end{proof}

\begin{theorem}\label{th:struct-QHMP}{\rm
In the notation above, let $\varphi_{H,O}\equiv (\varphi_{H,0}, \mathcal{E}_H,\mathcal{E}_{H,O})$ be a  Bi- entangled hidden Markov chain  with emission operator $\mathcal{E}_{H,O}$  given by \eqref{E_H_O}, a hidden transition expectation given by \eqref{E_H} and an  initial state $\varphi_{H,0}$ on $\mathcal{M}_{d}$,
then the joint expectations of the processes are given by

\begin{equation}\label{joint-exp-qHMP_entangled_1}
\varphi_{H,O}\Big(\bigotimes_{m=0}^{n}j_{H_{m}}(a_{m})\otimes j_{O_{m}}(b_{m})\Big)
\end{equation}

$$
=\varphi_{H_0}(a_0\diamond P_{H,O}(b_0)\diamond P_{H}(a_1\diamond P_{H,O}(b_1)\diamond P_{H}(a_2\diamond\cdots \diamond P_{H}(a_{n-1}\diamond P_{H,O}(b_{n-1})\diamond P_{H}(a_n\diamond P_{H,O}(b_n)\diamond P_{H}(\id_d)))\cdots )))
$$
Furthermore, the joint expectations of the processes can be expressed as follows
\begin{equation}\label{joint-exp-qHMP_entangled_2}
\varphi_{H,O}\Big(\bigotimes_{m=0}^{n}j_{H_{m}}(a_{m})\otimes j_{O_{m}}(b_{m})\Big)
\end{equation}
$$
=\sum_{\substack{i,j,h_{n+1}\\
o_{1},o_{1}',\cdots,o_{n+1},o_{n+1}'\\
l_1,k_1,\cdots,l_n,k_n}}a_{0,ij}b_{0,o_1,o_{1}'}a_{1,k_1l_1}b_{1,o_2,o_{2}'}\cdots a_{n-1,k_{n-1}l_{n-1}}b_{n-1,o_n,o_{n}'}a_{n,k_nl_n}b_{n-1,o_{n+1},o_{n+1}'}
$$
$$
\times\left( \sqrt{Q_{io_1}Q_{jo_1'}}\sqrt{\Pi_{ik_1}\Pi_{jl_1}} \sqrt{\prod_{m=1}^{n}Q_{k_mo_{m+1}}Q_{l_mo_{m+1}'}}\right) \left(\sqrt{\prod_{m=1}^{n-1}\Pi_{k_mk_{m+1}}\Pi_{l_ml_{m+1}}}\right)\sqrt{\Pi_{k_nh_{n+1}}\Pi_{l_nh_{n+1}}}\varphi_{H,0}(e_{ij})
$$

for all $n\in\mathbb{N}$,  for all $a_{m}= (a^{(m)}_{ij}), b_{m}=(b^{(m)}_{ij})\in\mathcal{M}_{d}, \,m\in \{0, \dots, n\}$.
}\end{theorem}
\begin{proof}
Since  $\mathcal{E}_{H,O}$ is identity preserving
\begin{eqnarray*}
\mathcal{E}_{H}(\mathcal{E}_{H,O}(a_n\otimes b_n)\otimes \id_{d}) &=&\mathcal{E}_{H,O}(a_n\otimes b_n)\diamond P_H(\id_{d})\\
&=& a_n \diamond P_{H,O}(b_n)\diamond P_H(\id_{d})\\
\end{eqnarray*}
Then
\begin{eqnarray*}
 \mathcal{E}_{H}(\mathcal{E}_{H,O}(a_{n-1}\otimes b_{n-1})\otimes \mathcal{E}_{H}(\mathcal{E}_{H,O}(a_n\otimes b_n)\otimes \id_{d}))  &=&\mathcal{E}_{H,O}(a_{n-1}\otimes b_{n-1})\diamond P_H(\mathcal{E}_{H}(\mathcal{E}_{H,O}(a_n\otimes b_n)\otimes \id_{d})))\\
&=& a_{n-1}\diamond P_{H,O}(b_{n-1})\diamond P_H(a_n \diamond P_{H,O}(b_n)\diamond P_H(\id_{d}))
\end{eqnarray*}
Iterating the above procedure then applying the initial state $\varphi_{H,0}$ we obtain \eqref{joint-exp-qHMP_entangled_1}.\\
Moreover, employing \eqref{joint-exp-qHMP_entangled_1}, it follows that
$$
a_n \diamond P_{H,O}(b_{n})\diamond P_{H}(\id_d)= \sum_{o_{n+1},o_{n+1}',h_{n+1}, i,j}a_{n;ij}b_{n;o_{n+1}o_{n+1}'}\sqrt{Q_{io_{n+1}}Q_{jo_{n+1}'}}\sqrt{\Pi_{ih_{n+1}}\Pi_{jh_{n+1}}} e_{ij}
$$
and
$$
P_{H}(a_n \diamond P_{H,O}(b_n)\diamond P_{H}(\id_d))=\sum_{\substack{k_n,l_{n},\\h_{n+1},o_{n+1},o_{n+1}' i,j}}a_{n,k_n l_n}b_{n;o_{n+1}o_{n+1}'}\sqrt{Q_{k_no_{n+1}}Q_{l_no_{n+1}'}}
\sqrt{\Pi_{i k_n}\Pi_{jl_{n}}\Pi_{k_n h_{n+1}}\Pi_{l_nh_{n+1}}}e_{ij}
$$
It follows that
$$
P_{H}(a_{n-1}\diamond P_{H,O}(b_{n-1})\diamond P_H (a_n \diamond P_{H,O}(b_{n})\diamond P_{H}(\id_d))))
$$
$$
=\sum_{\substack{k_{n-1},k_n,h_{n+1},\\l_{n-1},l_{n},o_{n+1},o_{n+1}',o_{n},o_{n}'\\i,j}}a_{n-1;k_{n-1}l_{n-1}}b_{n-1,o_n,o_{n}'}a_{n,k_nl_n}b_{n;o_{n+1}o_{n+1}'}\sqrt{Q_{k_{n-1}o_{n}}Q_{l_{n-1}o_{n}'}Q_{k_no_{n+1}}Q_{l_no_{n+1}'}}
$$
$$
\sqrt{\Pi_{ik_{n-1}}\Pi_{jl_{n-1}}\Pi_{k_{n-1}k_n}\Pi_{l_{n-1}l_n}\Pi_{k_{n}h_{n+1}}\Pi_{l_nh_{n+1}}}e_{ij}
$$
So by iteration,
$$
a_0\diamond P_{H,O}(b_0)\diamond P_{H}(a_1\diamond P_{H,O}(b_1)\diamond P_{H}(\cdots
\diamond P_{H}(a_{n-1}\diamond P_{H,O}(b_{n-1})\diamond P_H (a_n \diamond P_{H,O}(b_{n-1})\diamond P_{H}(\id_d)))\cdots)))
$$
$$
=\sum_{\substack{i,j,h_{n+1}\\
o_{1},o_{1}',\cdots,o_{n+1},o_{n+1}'\\
l_1,k_1,\cdots,l_n,k_n}}a_{0,ij}b_{0,o_1,o_{1}'}a_{1,k_1l_1}b_{1,o_2,o_{2}'}\cdots a_{n-1,k_{n-1}l_{n-1}}b_{n-1,o_n,o_{n}'}a_{n,k_nl_n}b_{n-1,o_{n+1},o_{n+1}'}\sqrt{Q_{io_1}Q_{jo_1'}}\sqrt{\Pi_{ik_1}\Pi_{jl_1}}
$$
$$
\times\left(\sqrt{\prod_{m=1}^{n}Q_{k_mo_{m+1}}Q_{l_mo_{m+1}'}}\right) \left(\sqrt{\prod_{m=1}^{n-1}\Pi_{k_mk_{m+1}}\Pi_{l_ml_{m+1}}}\right)\sqrt{\Pi_{k_nh_{n+1}}\Pi_{l_nh_{n+1}}}e_{ij}
$$

\end{proof}

\section{Recurrence of quantum Markov chains}
The main focus of this section is on the quantum Markov chain recurrence concepts.

\begin{definition}\label{stopping_time}
A (discrete) stopping time associated to a projection $e \in \mathcal{A}_H$ is a sequence $\{\tau_k\}_{k\geq0}$ with the following properties:
\begin{enumerate}[label=\alph*)]
  \item $\tau_k\in \mathcal{A}_{H,[0,k]}$, $\forall k\geq 0$;
  \item $\tau_k$ is a projection of $\mathcal{A}_H$, $\forall k\geq 0$;
  \item The $\tau_k$ are mutually orthogonal.
\end{enumerate}
\end{definition}
One can canonically associate a stopping time to any projection $e \in \mathcal{M}_d$ as follows:
$$
\begin{aligned}
\tau_{e; 0} & =e^{(0)} \otimes \id_{[1}= j_{H_0}(e) \\
\tau_{e; 1} & =e^{\perp} \otimes e \otimes \id_{[2}=j_{H_0}(e^{\perp})j_{H_1}(e) \\
 & \quad \quad \quad \vdots\\
\tau_{e; k} & =(e^{\perp})^{\otimes^k} \otimes e \otimes \id_{[k+1}=j_{H_0}(e^{\perp})\cdots j_{H_{k-1}}(e^{\perp})j_{H_k}(e).
\end{aligned}
$$
\begin{remark}
By relating projection $e$ to an event $E$, and interpreting index $n\in \mathbb{N}$ as a discrete time marker, the projection $\tau_k$ signifies the occurrence of event $E$ for the first time at moment $k$ (see \cite{AccDko}) . It's evident that the sequence $(\tau_k)$ meets the requirements of being a stopping time as defined in Definition \ref{stopping_time}.
\end{remark}

Put
$$
\tau_{e; n; \infty}:=j_{H_0}(e^{\perp})\cdots j_{H_{n-1}}(e^{\perp})j_{H_n}(e^{\perp});\quad \tau_{e; \infty}:=\lim_{n\rightarrow \infty}\tau_{\infty}^{n}=\bigotimes_{\mathbb{N}}e^{\perp}
$$
\textbf{Interpretation:}
The projection $\tau_{e; n; \infty}$ represents the scenario where event $E$ doesn't happen during the first $n$ instants, and the projection $\tau_{e; \infty}$ corresponds to the situation where event $E$ never occurs (see \cite{AccDko}).\\

Following \cite{AcFi-QMS}, we have the next result.

\begin{theorem}
Let $\varphi_H \equiv (\varphi_0, \mathcal{E})$ be a (homogeneous) quantum Markov chain on $\mathcal{A}_H$. There exists a unique
conditional expectation $E_{0]}$ from $\mathcal{A}_H$ into $\mathcal{M}_{d}$ characterized by
\begin{equation}\label{CE}
E_{0]}(a_0\otimes\cdots\otimes a_n)= \mathcal{E}(a_0\otimes \mathcal{E}(a_1\otimes \cdots\otimes \mathcal{E}(a_n\otimes \id))\cdots)
\end{equation}
for any $n\in \mathbb{N}$ and for all $a_0\otimes\cdots\otimes a_n\in \mathcal{A}_{H,[0,n]}$. Furthermore, one has
\begin{equation}\label{E_0_CE_varphi}
  \varphi_{H}(.)= \varphi_{0}\circ E_{0]}
\end{equation}
\end{theorem}
\begin{definition}
 Let $\varphi_H \equiv\left(\varphi_{0}, \mathcal{E}\right)$ be a (homogeneous) quantum Markov chain on $\mathcal{A}_H$. A projection $e \in$ $\mathcal{M}_d$ is said to be
\begin{enumerate}
\item $\mathcal{E}$-completely accessible if
\begin{equation}\label{
E_accessible}
E_{o]}\left(\tau_{e,\infty}\right):=\lim _{n \rightarrow \infty} E_{o]}\left(\tau_{e ; n; \infty}\right)=0
\end{equation}
\item $\varphi_H$-completely accessible if
\begin{equation}\label{varphi_accessible}
\varphi_H\left(\tau_{e ; \infty}\right)=0
\end{equation}
\item $\mathcal{E}$-recurrent if $0<\operatorname{Tr}(\mathcal{E}(e \otimes \mathbb{I}))<\infty$ and one has

\begin{equation}\label{E_recurrent}
\frac{1}{\operatorname{Tr}(\mathcal{E}(e \otimes \mathbb{I}))} \operatorname{Tr}\left(E_{o]}\left(\sum_{n \geq 0} e \otimes \tau_{e ;n}\right)\right)=1.
\end{equation}
\item $\varphi_H$-recurrent if $\varphi_H\left(j_{H_0}(e)\right) \neq 0$ and

\begin{equation}\label{varphi_recurrent}
\frac{1}{\varphi_H\left(j_{H_0}(e)\right)} \varphi_H\left(\sum_{n} e \otimes \tau_{e ; n}\right)=1
\end{equation}
\end{enumerate}
\end{definition}

\begin{definition}
Let $\varphi_H \equiv\left(\varphi_{0}, \mathcal{E}\right)$ be (homogeneous) a quantum Markov chain. Let $e, f \in \operatorname{Proj}(\mathcal{M}_d), e, f \neq$ 0. The projection $f$ is
\begin{enumerate}
  \item $\mathcal{E}$-accessible from $e$ (and we write $e \rightarrow^{\mathcal{E}} f$ ) if there exists $m \in \mathbb{N}$ such that
$$
E_{o]}\left(j_{H_0}(e) j_{H_m}(f)\right) \neq 0
$$
If $e\rightarrow^{\mathcal{E}}  f$ and $f \rightarrow^{\mathcal{E}}  e$, then we sat that $e$ and $f$ $\mathcal{E}$- communicate and we write $e\leftrightarrow^{\mathcal{E}}  f$.
\item $\varphi_H$-accessible from $e$ (we denote it as $e \rightarrow^{\varphi_H} f$ if there exists $m \in \mathbb{N}$ such that
$$
\varphi_H\left(j_{H_0}(e) j_{H_{m}}(f)\right) \neq 0
$$
If $e\rightarrow^{\varphi_H}  f$ and $f \rightarrow^{\varphi_H}  e$, then we sat that $e$ and $f$ $\varphi_H$- communicate and we write $e\leftrightarrow^{\varphi_H}  f$.
\end{enumerate}
\end{definition}

\begin{lemma}\label{lemma_identity}
In the notations provided earlier:

\begin{equation}\label{Sum_stopping_time}
\sum_{n \geq 0} \tau_{e ; n}=\id_{\mathcal{A}_{H}}-\tau_{e; \infty}
\end{equation}
Moreover, a projection $e$ is $\mathcal{E}$- completely accessible if and only if
\begin{equation}\label{E_completely_accessible}
  E_{0]}\left(\sum_{n}\tau_{e;n}\right)=\id_d
\end{equation}
\end{lemma}
\begin{proof}
It can be observed that
$$
\tau_{e;0}+\tau_{e;1}= \id_{\mathcal{A}_H}- j_{H_0}(e^{\perp})j_{H_1}(e^{\perp})
$$
Iterating this procedure, it becomes clear that
\begin{equation}\label{iterating_stopping}
  \tau_{e;0}+\tau_{e;1}+\cdots+\tau_{e;n}= \id_{\mathcal{A}_H}-j_{H_0}(e^{\perp})j_{H_1}(e^{\perp})\cdots j_{H_n}(e^{\perp})=  \id_{\mathcal{A}_H}-\tau_{e,n,\infty}
\end{equation} So by taking $n\rightarrow \infty$, we get \eqref{Sum_stopping_time}. Now, by considering the expectation $E_{0]}$ for both sides of \eqref{iterating_stopping} and taking the limit as $n\rightarrow \infty$, this leads us to \eqref{E_completely_accessible}.
\end{proof}

\begin{theorem}
Let $\varphi_H \equiv\left(\varphi_{0}, \mathcal{E}\right)$ be a (homogeneous) quantum Markov chain on $\mathcal{A}_{H}$. Let $e \in \mathcal{M}_d$ be a projection
\begin{enumerate}[label=\alph*)]
\item $e$ is $\mathcal{E}$-recurrent if and only if

\begin{equation}\label{E_recuurent_condition}
\mathcal{E}\left(e \otimes E_{o]}\left(\tau_{e ; \infty}\right)\right)=0
\end{equation}

\item  $e$ is $\varphi_H$-recurrent if and only if one has

\begin{equation}\label{varphi_recurrent_condition}
\varphi_H\left(e \otimes \tau_{e ; \infty}\right)=0
\end{equation}

\end{enumerate}
\end{theorem}

\begin{proof}
 We shall use the following necessary and sufficient condition for $e$ to be $\mathcal{E}$- recurrent (Proposition 2.1, \cite{AccDko}):
$$
  E_{0]}(e\otimes \sum_{n}\tau_{e,n})= E_{0]}(e)
$$
and  As a consequence of Lemma \ref{lemma_identity}, we can establish:

$$
\sum_{n \geq 0} e \otimes \tau_{e, n}=e \otimes\id_{\mathcal{A}_H} -e \otimes \tau_{e; \infty}
$$
As a result, we obtain both (i) and (ii).

\end{proof}

\begin{corollary}\label{E_reccurent_varphi}
Let $\varphi_H \equiv\left(\varphi_{0}, \mathcal{E}\right)$ be a (homogeneous) quantum Markov chain.\\ Every $\mathcal{E}$-recurrence projection implies $\varphi_H$-recurrence, and conversely, when the initial state $\varphi_{0}$ is faithful, every $\varphi_H$-recurrence projection implies $\mathcal{E}$-recurrence.
\end{corollary}

\begin{proof}
-- Let $e \in \mathcal{A}_H$ be a projection. From \eqref{E_0_CE_varphi}, one has
$$
\begin{aligned}
\left.\varphi_H\left(e \otimes \tau_{e; \infty}\right)\right) = & \varphi_{0}\left(E_{0]}\left(e \otimes \tau_{e ; \infty}\right)\right) \\
= & \varphi_{0}\left(\mathcal{E}\left(e \otimes E_{0]}\left(\tau_{e ; \infty}\right)\right)\right)
\end{aligned}
$$
Consequently, if $\mathcal{E}\left(e \otimes E_{0]}\left(\tau_{e ;\infty}\right)\right)=0$, then $\varphi_H\left(e \otimes \tau_{e ;\infty}\right)=0$, demonstrating the first implication.\\
-- Assuming the initial state $\varphi_{0}$ is faithful, and given that $\mathcal{E}\left(e \otimes E_{0]}\left(\tau_{e ;\infty}\right)\right) \geq 0$, the preceding calculation leads us to:
$$
\varphi_H\left(e \otimes \tau_{e; \infty}\right)=0 \Rightarrow \mathcal{E}\left(e \otimes E_{0]}\left(\tau_{e ;\infty}\right)\right)=0
$$

This establishes the inverse implication, concluding the proof.
\end{proof}

\section{Recurrence for underlying Markov process associated to Bi- entangled hidden Markov models}

In this section, we focus on the concept of recurrence for the quantum Markov chain $\varphi_H^{(O)}$ on $\mathcal{A}_H$, which is derived from the restriction of the bi-entangled hidden Markov models $\varphi_{H,O}$ on the algebra $\mathcal{A}_H$.
\begin{equation}\label{phiH(O)}
\varphi_H^{(O)} := \varphi_{H,O}\lceil_{\mathcal{A}_H}
\end{equation}
In \cite{AccSegLuSs}, the quantum Markov chain defined by (\ref{phiH(O)}) is defined to be the underlying (Hidden) Markov process associated with the HQMM  $\varphi_{H,O}$.\\
In certain significant cases, the two Quantum Markov Chains (QMCs), denoted as $\varphi_{H}\equiv (\varphi_{H,0}, \mathcal{E}_{H})$ and $\varphi_{H}^{(O)}\equiv (\varphi_{H,0}, \mathcal{E}_{H}^{(O)})$, are equivalent. This equivalence holds particularly when the restriction of the map $\mathcal{E}_{H,O}$ to $\mathcal{M}_{d}$ is equal to the identity map.\\
Let define the map $\mathcal{E}^{(O)}_{H}$ as follows
\begin{equation}\label{Restric_TE}
\mathcal{E}^{(O)}_{H}(a\otimes b) = \mathcal{E}_{H}(\mathcal{E}_{H,O}(a\otimes \id_{d})\otimes b), \quad \forall a, b\in\mathcal{M}_{d}
\end{equation}
\begin{theorem}
In the notation above, the map $\mathcal{E}^{(O)}_{H}$ define a Markov transition expectation (i.e a complelety positive identity preserving map) from $\mathcal{M}_{d}\otimes \mathcal{M}_{d}$ to $\mathcal{M}_{d}$.\\The expression of  $\mathcal{E}^{(O)}_{H}$ is as follow
$$
\mathcal{E}^{(O)}_{H}(a\otimes b)=a \diamond P_{H,O}(\id)\diamond P_{H}(b)
$$
Furthermore, the backward Markov operators corresponding to $\mathcal{E}^{(O)}_{H}$ can be expressed as follows
\begin{equation}\label{E_O_H_underlying}
\mathcal{E}^{(O)}_{H}(a\otimes b)=\sum_{k,m,l,i,j}a_{ij}b_{ml}\sqrt{Q_{ik}Q_{jk}}\sqrt{\Pi_{im}\Pi_{jl}}e_{ij}
\end{equation}
for all $a=(a_{ij})$, $b=(b_{ij})$ $\in \mathcal{M}_d$.
\end{theorem}
\begin{proof}
Due to its construction  $\mathcal{E}^{(O)}_{H}$ as composition of completely positive identity preserving maps,  $\mathcal{E}^{(O)}_{H}$ is a completely positive identity preserving map. Furthermore,
\begin{eqnarray*}
\mathcal{E}^{(O)}_{H}(a\otimes b) &=& \mathcal{E}_{H}(\mathcal{E}_{H,O}(a\otimes \id_{d})\otimes b) \\
   &=& \mathcal{E}_{H,O}(a\otimes \id_{d}) \diamond P_H(b) \\
   &=& a\diamond P_{H,O}(\id_d)\diamond P_{H}(b) \\
   &=& \left(\sum_{i,j,k}a_{ij}\sqrt{Q_{ik}Q_{jk}e_{ij}}\right)\diamond \left(\sum_{i,j,m,l}b_{ml}\sqrt{\Pi_{im}\Pi_{jl}}e_{ij}\right)\\
   &=& \sum_{i,j,k,m,l}a_{ij}b_{ml}\sqrt{Q_{ik}Q_{jk}}\sqrt{\Pi_{im}\Pi_{jl}}e_{ij}
\end{eqnarray*}
\end{proof}

\begin{theorem}
In the above notations, the conditional expectation $E_{H,0]}^{(O)}$ associated with $\mathcal{E}^{(O)}_{H}$
through \eqref{CE} has the following expression
\begin{equation}\label{CE_REcurrence}
E_{H,0]}^{(O)}(a_0\otimes \cdots\otimes a_n)
\end{equation}
$$
=a_0\diamond P_{H,O}(\id_d)\diamond P_{H}(a_1\diamond P_{H,O}(\id_d)\diamond P_{H}(\cdots
\diamond P_{H}(a_{n-1}\diamond P_{H,O}(\id_d)\diamond P_H (a_n \diamond P_{H,O}(\id_d)\diamond P_{H}(\id_d)))\cdots)))
$$
\begin{equation}\label{E_0_bi_entangled}
=\sum_{\substack{i,j,h_{n+1}\\
o_1,\cdots,o_{m+1}\\
l_1,k_1,\cdots,l_n,k_n}}a_{0,ij}a_{1,k_1l_1}\cdots a_{n-1,k_{n-1}l_{n-1}}a_{n,k_nl_n}\sqrt{Q_{io_1}Q_{jo_1}}\sqrt{\Pi_{ik_1}\Pi_{jl_1}}\times\left(\sqrt{\prod_{m=1}^{n}Q_{k_mo_{m+1}}Q_{l_mo_{m+1}}}\right)
\end{equation}

$$
\times \left(\sqrt{\prod_{m=1}^{n-1}\Pi_{k_mk_{m+1}}\Pi_{l_ml_{m+1}}}\right)\sqrt{\Pi_{k_nh_{n+1}}\Pi_{l_nh_{n+1}}}e_{ij}
$$
\end{theorem}

\begin{proof} The identity \eqref{E_0_bi_entangled} is obtained replacing in the right hand side of \eqref{joint-exp-qHMP_entangled_2} the $b_m$ by $\id_d$.
\end{proof}
Let  $\varphi_{H,0}$ be a initial state on $\mathcal{M}_{d}$ such that
\begin{equation}\label{intial_state}
\varphi_{H,0}(\cdot)=\hbox{Tr}(W_0 \cdot); \quad W_0 = \sum_{j\in D_H}\pi_j e_{jj}
\end{equation}
where $\pi =(\pi_j)_{j\in D_H }$ the initial distribution of the hidden process.
\begin{theorem}\label{recurrent_theorem}
In the notation above, if $e=(e_{p;ij})_{i,j\in D}$ is a projection in $\mathcal{A}_H$ such that
\begin{equation}\label{e_projection}
  q:= \sum_{i,j}e_{p;ij}^{\perp}< \frac{1}{d}
\end{equation}
then $e$ is $\varphi_H^{(O)}$-recurrent.
\end{theorem}
\begin{proof}
From \eqref{E_0_CE_varphi}, we have that
$$
\varphi_{H}^{(O)}(e\otimes \tau_{e,n,\infty}) =\varphi_{H,0}(E_{H,0]}^{(O)}(e\otimes \tau_{e,n,\infty}))
$$
$$
=\sum_{\substack{
o_{2},\cdots,o_{n+1}, h_{n+1}\\
l_2,k_2,\cdots,l_n,k_n}}\hbox{Tr}\left(W_0 \left(\sum_{i,j,o_1,l_1,k_1}e_{p,ij}\sqrt{Q_{io_1}Q_{jo_1}}\sqrt{\Pi_{ik_1}\Pi_{jl_1}}e_{ij}\right)\right)\times e^{\perp}_{p,k_1l_1}\cdots e^{\perp}_{p,k_{n-1}l_{n-1}}e^{\perp}_{p,k_nl_n}
$$
$$
\times \left(\sqrt{\prod_{m=1}^{n}Q_{k_mo_{m+1}}Q_{l_mo_{m+1}}}\right)\left(\sqrt{\prod_{m=1}^{n-1}\Pi_{k_mk_{m+1}}\Pi_{l_ml_{m+1}}}\right)\sqrt{\Pi_{k_nh_{n+1}}\Pi_{l_nh_{n+1}}}
$$
$$
=\sum_{\substack{
i,o_{1},\cdots,o_{n+1}, h_{n+1}\\
l_1,k_1,\cdots,l_n,k_n}}p_i e_{p,ii}Q_{io_1}\sqrt{\Pi_{ik_1}\Pi_{il_1}}\times e^{\perp}_{p,k_1l_1}\cdots e^{\perp}_{p,k_{n-1}l_{n-1}}e^{\perp}_{p,k_nl_n}
$$
$$
\times \left(\sqrt{\prod_{m=1}^{n}Q_{k_mo_{m+1}}Q_{l_mo_{m+1}}}\right)\left(\sqrt{\prod_{m=1}^{n-1}\Pi_{k_mk_{m+1}}\Pi_{l_ml_{m+1}}}\right)\sqrt{\Pi_{k_nh_{n+1}}\Pi_{l_nh_{n+1}}}
$$
$$
\leq \left(\sum_{i,o_1}e_{p,ii}\right)\left(\sum_{o_2,k_1,l_1}e^{\perp}_{p,k_1l_1}\right)\cdots \left(\sum_{o_{n+1},k_n,l_n}e^{\perp}_{p,k_nl_n}\right)\leq d^2\times (dq)^{n-1}
$$
Thus, $\varphi_{H}^{(O)}(e\otimes \tau_{e,\infty})=0$ and by \eqref{varphi_recurrent_condition} the projection $e$ is $\varphi_H^{(O)}$-recurrent.
\end{proof}
\begin{remark}Given the faithfulness of the initial state $\varphi_{H,0}$, corollary \ref{E_reccurent_varphi} implies that the projection $e$ is also $\mathcal{E}_{H}^{(O)}$-recurrent.
\end{remark}
\section{Diagonal restriction of the underlying Markov process associated to Bi- entangled hidden Markov models}

Within this section, we demonstrate that the restriction of the underlying Markov process associated to Bi- entangled hidden Markov models
on any diagonal algebra is a classical Markov process.This process is defined by the hidden stochastic matrix $\Pi=(\Pi_{ij})_{i,j\in D}$ and the initial distribution $\pi =(\pi_j)_{j\in D}$.\\

In the following, we denote $\mathcal{D}_{e}$ the $e$-diagonal sub- algebra of $\mathcal{M}_d$ defined by
\begin{equation}\label{D_e}
  \mathcal{D}_e:=\{\sum_{h\in D}x_h e_{h,h}: x_h\in \mathbb{C}\}
\end{equation}
The $e$- diagonal sub- algebra of $\mathcal{A}_H$ is defined by
$$
\mathcal{D}_H:=\bigotimes_{\mathbb{N}}\mathcal{D}_e
$$
\begin{remark}
\begin{equation}\label{df-De}
\mathcal{D}_H:=\bigotimes_{\mathbb{N}}\mathcal{D}_e\equiv \bigotimes_{\mathbb{N}}L^{\infty}_{\mathbb{C}}(D)
\end{equation}
where $L^{\infty}_{\mathbb{C}}(D)$ denotes the space of all functions $f:$ $D\rightarrow \mathbb{C}$.\\
Therefore, if $\varphi_H\equiv(\varphi_0,\mathcal{E})$ is any quantum Markov chains on $\mathcal{A}_H$,
for any diagonal algebra $\mathcal{D}_{e}$, the restriction of $\varphi_H$ on $\mathcal{D}_{H}$
defines a unique classical process $H\equiv (H_n)$, with state space $D$, characterized by the joint probabilities
$$
\hbox{Prob}\Bigl(H_0=i_0, H_1=i_1,\cdots ,H_n=i_n\Bigr)=
\varphi_H\bigl(e_{i_0i_0}\otimes e_{i_1i_1}\otimes \cdots \otimes
e_{i_{n}i_{n}}\bigr)=
$$
\begin{equation}\label{(2.1)}
=\varphi_0\left(\mathcal{E}(e_{i_0i_0}\otimes \mathcal{E}(e_{i_1i_1}\otimes
\cdots \otimes \mathcal{E}(e_{i_{n-1}i_{n-1}}\otimes
\mathcal{E}(e_{i_ni_n}\otimes 1))\cdots ))\right)
\end{equation}
for any $n\in\mathbb{N}$, $\{i_h\}_{h=0}^n\subset \{1,\cdots, d\}$.\\
With the identification \eqref{df-De} the restriction of the  embedding $j_{H_n}$,
defined by \eqref{df-H-embeds-cl}, to $\mathcal{D}_{e}$ can be identified to
\begin{equation}\label{(1.6)}
j_{H_n}(f):=f(H_n),\ \ \ \forall\ f\in L^{\infty}_{\mathbb{C}}(D)
\end{equation}
\end{remark}

\begin{lemma}\label{trans-exp-diagonal-restriction}
The  transition expectation  $\mathcal{E}_{H}^{(O)}$ \eqref{Restric_TE} maps the diagonal algebra
$\mathcal{D}_{e}\otimes\mathcal{D}_{e}$ into the diagonal algebra $\mathcal{D}_{e}$.
\end{lemma}
\begin{proof}
Let  $a = \sum_{i\in D}x_ie_{ii},\,  b= \sum_{i\in D}y_ie_{ii} \in\mathcal{D}_{e}$, from \eqref{E_O_H_underlying}, it follows that:
\begin{eqnarray*}
  \mathcal{E}^{(O)}_{H}(a\otimes b) &=& \mathcal{E}_{H}(\mathcal{E}_{H,O}(a\otimes \id_{d})\otimes b) \\
   &=& \sum_{k,l,i}x_{i}y_{l}Q_{ik}\Pi_{il}e_{ii}.
\end{eqnarray*}
\end{proof}
\begin{theorem}
In the notations of theorem \ref{recurrent_theorem}, the restriction of $\varphi_{H}^{(O)}$ to the diagonal algebra
$\mathcal{D}_H $ is characterized by the joint probabilities
\begin{equation}\label{joint-exps-Diagonal-restrict}
\varphi_{H}^{(O)}\left(\prod_{m=0}^{n}j_{H_{m}}(e_{j_mj_m}) \right)=\pi_{j_0} \prod_{m=0}^{n-1}\Pi_{j_m j_{m+1}}
\end{equation}
where $k_0,\dots, k_n\in D$.
\end{theorem}
\begin{proof}
From \eqref{E_0_CE_varphi}, we have that
$$
\varphi_{H}^{(O)}(\prod_{m=0}^{n}j_{H_{m}}(e_{j_mj_m})) =\varphi_{H,0}(E_{H,0]}^{(O)}(\prod_{m=0}^{n}j_{H_{m}}(e_{j_mj_m})))
$$
$$
\begin{aligned}
&=\sum_{o_{2},\cdots,o_{n+1}, h_{n+1}}\hbox{Tr}\left(W_0 \left(\sum_{o_1}Q_{j_0o_1}\sqrt{\Pi_{j_0j_1}\Pi_{j_0j_1}}e_{j_0j_0}\right)\right)\times \left(\prod_{m=1}^{n}Q_{j_mo_{m+1}}\right)\left(\prod_{m=1}^{n-1}\Pi_{j_mj_{m+1}}\right)\Pi_{j_nh_{n+1}}\\
&=\sum_{o_{1},\cdots,o_{n+1}, h_{n+1}}p_{j_0} \Pi_{j_0j_1}\Pi_{j_0j_1}\times \left(\prod_{m=1}^{n}Q_{j_mo_{m+1}}\right)\left(\prod_{m=1}^{n-1}\Pi_{j_mj_{m+1}}\right)\Pi_{j_nh_{n+1}}\\
&=\pi_{j_0} \prod_{m=0}^{n-1}\Pi_{j_m j_{m+1}}
\end{aligned}
$$

\end{proof}
\begin{remark}
It is readily apparent that the joint probabilities \eqref{joint-exps-Diagonal-restrict} gives the joint  probabilities of the  Markov defined by the hidden stochastic matrix $\Pi=(\Pi_{ij})_{i,j\in D}$ and the initial distribution $\pi =(\pi_j)_{j\in D}$.
\end{remark}
\begin{remark}
This shows that the  underlying Markov process associated to Bi- entangled hidden Markov models doesn't belongs to the special class of quantum Markov chains that is
strictly related to classical hidden Markov processes in the sense that: the
restriction of any element in this class to any diagonal (in particular commutative)
sub- algebra produces a hidden Markov processes.
\end{remark}
\begin{remark}
In \cite{AccSegLuSs}, it was shown that any diagonalizable quantum Markov chain belongs to this special class. This fact has been recognized since the very beginning of the quantum Markov chains theory. (see \cite{Ac91-Q-Kalman-filters}).
\end{remark}


\begin{thebibliography}{99}


\bibitem{Ac74-Camerino}
L. Accardi: {\it Non--commutative Markov chains},
Proceedings International School of Mathematical Physics,
Universit\`a di Camerino 30 Sept., 12 Oct. p.268-295 (1974)

\bibitem{Ac91-Q-Kalman-filters}
Accardi L.:\\
Quantum Kalman filters,
in: Mathematical system theory,\\
The influence of R.E. Kalman, A.C. Antoulas (ed.)
Springer (1991) 135-143\\
Invited contribution to the memorial volume for
the 60-th birthday of R.E.Kalman\\





\bibitem{AccDko}
L. Accardi, D. Koroliuk, \it{Quantum Markov Chains: The recurrency Problem }, QP III World Scientific.
\bibitem{AcFi03-EMC}
L. Accardi, F. Fidaleo, {\it Entangled Markov Chains},
Annali di Matematica Pura e Applicata, 184 (3), p.327--346 (2005).

\bibitem{AcMaOh06}
Accardi, L., Matsuoka, T., Ohya, M., {\it Entangled Markov chains are indeed entangled}. Infinite Dimensional Analysis, Quantum Probability and Related Topics, 9(03), pp.379-390 (2006)

\bibitem{SsSeBa}
A. Souissi, El. Soueidy and A. Barhoumi, {\it On a $\psi$-Mixing property for Entangled Markov Chains}. Physica A: Statistical Mechanics and its Application, 613, pp.128-533 (2023)


\bibitem{SsSe}
A. Souissi, El. Soueidi, {\it Entangled Hidden Markov Models}. Chaos, Solitons and Fractals, 174, pp.113- 804 (2023)

\bibitem{AcFi-QMS}
Accardi L., Fidaleo F., \it{Non homogeneous quantum Markov states and quantum Markov fields},
J. Funct. Anal. 200 (2003), 324-–347.
\bibitem{LaOhnFm} Accardi, L., Ohno, H., Mukhamedov, F.: {\it Quantum Markov fields on graphs.} Inf. Dim. Anal. Quantum Probab. Relat. Top. 13, 165–189 (2010)

\bibitem{AcSouElG20}
L. Accardi, A. Souissi, El. Soueidy: {\it Quantum Markov chains, A unification approach}, Infin. Dimens. Anal. Quantum Probab. Relat. Top. (IDAQP)  23 (2) (2020)


\bibitem{AccSegLuSs} Accardi L, Soueidi EG, Lu YG, Souissi A. A hidden processes and hidden Markov
processes: classical and quantum. 2023, arXiv preprint arXiv:2302.07058.
\bibitem{Algh2016}
R. Alghamdi:
{\it Hidden Markov models (HMMs) and security applications}.
Int. J. Adv. Comput. Sci. Appl. 7 (2), p.39--47 (2016)


\bibitem{BaumPetr66}
L.E. Baum, T. Petrie:
{\it Statistical inference for probabilistic functions of finite state Markov chains}.
The Annals of Mathematical Statistics, 37, p.1554-1563 (1966)

\bibitem{CGGK2017}
M. Cholewa, P. Gawron, P. Glomb, D. Kurzyk:
{\it Quantum hidden Markov models based on transition operation matrices}.
Quantum Information Processing, 16 (4), p.1--19 (2017)
\bibitem{Dhari-Farrukh} Dhahri A., Mukhamedov F., {\it Open quantum random walks, quantum Markov chains and
recurrence}, Rev. Math. Phys. 31 (2019), no. 7, 1950020.
\bibitem{Eddy98}
S.R. Eddy: {\it Profile hidden Markov models}. Bioinformatics, 14 (9),  p.755-763 (1998)

\bibitem{JasKell12} J. Ernst, M. Kellis:
{\it ChromHMM: automating chromatin-state discovery and characterization}. Nature methods, 28, 9, p.215-216 (2012)



\bibitem{FelsChur92}
J. Felsenstein, G.A. Churchill:
{\it A Hidden Markov Model approach to variation among sites in rate of evolution}.
 Molecular biology and evolution, 13 (1), p.93-104 (1996)

\bibitem{GhahrJord97}
Z. Ghahramani, M.I. Jordan:
{\it Factorial Hidden Markov Models}. Machine Learning, 29,  p. 245-273 (1997)

\bibitem{HuaYasMerv90}
X.D. Huang, A. Yasuo, M. Jack:
{\bf Hidden Markov Models for Speech Recognition},  Columbia University Press (1990) ISBN:978-0-7486-0162-2

\bibitem{JelBahMer75}
F. Jelinek, L. Bahl, R. Mercer:
 {\it Design of a linguistic statistical decoder for the recognition of continuous speech}.
IEEE Transactions on Information Theory, 21 (3), p.250-256 (1975)

\bibitem{HasNat05}
M.R. Hassan, B. Nath,
{\it Stock market forecasting using hidden Markov model: a new approach},
5th International Conference on Intelligent Systems Design and Applications (ISDA'05), p.192-196 (2005)

\bibitem{LiSte03}
Na Li, M. Stephens:
{\it Modeling linkage disequilibrium and identifying recombination hotspots using single-nucleotide polymorphism data}. Genetics, 165 (4),  p.2213-33 (2003), doi: 10.1093/genetics/165.4.2213



\bibitem{MGK2021}
B. Mor, S. Garhwal, A. Kumar:
{\it A systematic review of hidden markov models and their applications}. Archives of computational methods in engineering, 28 (3), p.1429-1448 (2021)

\bibitem{MonrWiesn11} A. Monras, A. Beige, and K. Wiesner: {\it Hidden Quantum Markov Models and non-adaptive read-out of many-body states},
App. Math. Comput. Sci. 3, 93 (2011)

\bibitem{Nguyet18} N. Nguyen: {\it Hidden Markov Model for Stock Trading}.
International Journal of Financial Studies, 6 (2), 36 (2018) https://doi.org/10.3390/ijfs6020036



\bibitem{RosPent98} N.M. Oliver, B. Rosario, A. Pentland: {\it Graphical models for recognizing human interactions}. Advances in Neural Information Processing Systems, 11, p.24-30  (1998)

\bibitem{PardBirm05}B. Pardo, W. Birmingham.
{\it Modeling Form for On-line Following of Musical Performances}. AAAI'05: Proceedings of the 20th national conference on Artificial intelligence, v.2, p.1018-1023 (2005)

\bibitem{Rab88} L.R. Rabiner:
 {\it Mathematical foundations of hidden Markov models}.
 In: Recent advances in speech understanding and dialog systems,
 Springer, Berlin, Heidelberg, p.183-205 (1988)

\bibitem{Rab89} L.R. Rabiner,
{\it A tutorial on hidden Markov models and selected applications in speech recognition},
Proceedings of the IEEE 77 (2), p.257--286 (1989)


\bibitem{RabJua86}L.R. Rabiner, B.H. Juang:
{\it An introduction to hidden Markov models}.
IEEE ASSP magazine, 3 (1), p.4-16 (1986)

\bibitem{RabLeSo83}
L.R. Rabiner, S.E. Levinsion, M.M. Sondhi,
{\it On the Application of Vector Quantization and Hidden Markov Models to Speaker-Independent Isolated Word Recognition},
Bell System Tech. J., v.62, n.4, p.1075--1105 (1983)



\bibitem{RebSa17}
Sara Rebagliati, Emanuela Sasso\\
Pattern recognition using hidden Markov models in financial time series\\
ACTA ET COMMENTATIONES UNIVERSITATIS TARTUENSIS DE MATHEMATICA, (21) 1 (2017) 1--17\\
http://acutm.math.ut.ee

\bibitem{SsSeg}
Souissi A., and  Soueidi EG., {\it Entangled Hidden Markov Models},
Chaos, Solitons and Fractals, 174, p.113--804 (2023)

\bibitem{SatGuru93}L. Satish, B.I. Gururaj,
 {\it Use of hidden Markov models for partial discharge pattern classification}, IEEE Transactions on Electrical Insulation, 28 (2), p.172-182 (1993)

\bibitem{SSGGBB18}S. Srinivasan, G. Gordon, B. Boots: {\it Learning hidden quantum Markov models}.
Proceedings of the 21st International Conference on Artificial Intelligence and Statistics (AISTATS) 2018, PMLR: v.84, p.1979-1987

\bibitem{WiesnCrutc08} K. Wiesner, C.P. Crutchfield:
 {\it Computation in finitary stochastic and quantum processes},
Physica D, v.237, iss.9, p.1173-1195 (2008), https://doi.org/10.1016/j.physd.2008.01.021

\bibitem{YamOhyIsh92}J.
Yamato, J. Ohya, K. Ishii:
 {\it Recognizing human action in time-sequential images using hidden Markov model}.
 Proceedings 1992 IEEE Computer Society Conference on Computer Vision and Pattern Recognition (CVPR), p. 379-385, doi:10.1109/CVPR.1992


\end{thebibliography}
\end{document}